\documentclass[11pt]{article}

\usepackage[T1]{fontenc}
\usepackage[utf8]{inputenc}
\usepackage[margin=1.05in]{geometry}
\usepackage{lmodern}
\usepackage{microtype}
\usepackage{amsmath,amssymb,amsthm,mathtools}
\usepackage{enumitem,booktabs,array}
\usepackage{float}
\usepackage{natbib}
\usepackage{xcolor}
\usepackage{xurl}
\usepackage{hyperref}
\hypersetup{colorlinks=true,linkcolor=blue!50!black,citecolor=blue!50!black,urlcolor=blue!50!black}

\numberwithin{equation}{section}

\newtheorem{assumption}{Assumption}
\newtheorem{theorem}{Theorem}
\newtheorem{lemma}{Lemma}
\newtheorem{proposition}{Proposition}

\theoremstyle{remark}
\newtheorem{remark}{Remark}

\newcommand{\R}{\mathbb{R}}
\newcommand{\one}{\mathbf{1}}

\newcommand{\unif}{\mathrm{unif}}

\setlist[enumerate]{itemsep=0.25em,topsep=0.25em}
\setlist[itemize]{itemsep=0.25em,topsep=0.25em}

\title{Aggregating Elo Ratings: An Axiomatization\thanks{ Disclaimer. The author leads a working group responsible for the rules and structures for the Total Chess World Championship Tour. Discussions within the group, involving Norway Chess and FIDE, encouraged the author to develop a proposal for a combined rating system. This preliminary draft is released for public feedback, does not represent an official position or endorsement of Norway Chess or FIDE, and reflects the author's own views and proposals. Comments and constructive criticism are welcome at mehmet.mars.seven@kcl.ac.uk.}}
\author{Mehmet Mars Seven \\ King's College London, London, UK.}
\date{\today}

\begin{document}

\maketitle

\begin{abstract}
\noindent
Many environments assign several Elo ratings to the same agent: a chess player has classical, rapid, and blitz ratings; an online platform may rate by time control, mode, or format; an evaluator may rate performance across tasks or roles. This paper axiomatizes when such a vector of ratings can be reduced to a single scalar rating that is itself on the Elo scale. We impose three substantive conditions: same-scale normalization (a uniform profile keeps its rating), recursive consistency (aggregating in blocks gives the same answer as aggregating directly, provided each block carries the total weight of its members), and marginal Elo-strength consistency (for two equally weighted coordinates, the ratio of marginal contributions to the combined rating equals the ordinary Elo odds). The unique rating rule satisfying these conditions converts each component to its Elo strength, takes a weighted arithmetic mean of strengths, and converts back. We show how this rule differs from a random-format lottery and from rating-scale averaging, prove the axioms are independent, and illustrate the rule on combining classical, rapid, and blitz ratings. 
\end{abstract}



\section{Introduction}\label{sec:introduction}

Many competitive and evaluative systems assign several ratings to the same agent. A chess player may have separate classical, rapid, and blitz ratings; an online player may be rated by time control or mode; and an evaluator may maintain task-specific or role-specific ratings. This raises a natural aggregation question: when can a vector of Elo ratings be summarized by a single number that is itself interpretable as an Elo rating?

This paper studies that question for an aggregated rating: a scalar Elo rating intended to measure aggregate strength across several coordinates. In the chess application, classical, rapid, and blitz ratings are treated as components of a multi-format strength profile, and the output is a single Elo-style measure of combined strength.

The distinction between such a rating and other uses of multi-format information is important. If the goal is to predict a game in a known format, the relevant object is the format-specific rating. If the goal is to evaluate a competition in which the format is randomly selected, the natural object is a weighted average of format-specific win probabilities. The object studied here is different: it asks for one Elo-scale measure of overall multidimensional strength.

A simple candidate is the weighted arithmetic average of the ratings. But Elo ratings are logarithmic. A 400-point difference represents a tenfold ratio of underlying strengths. Averaging directly on the rating scale therefore treats a one-point gain at rating 1200 as having the same marginal force as a one-point gain at rating 2800. If the output is to remain an Elo rating, the natural aggregation scale is instead the underlying Elo-strength scale.

The paper gives an axiomatic characterization of this idea. Let
\[
R=(R_1,\ldots,R_n)
\]
be a rating profile and let
\[
\lambda=(\lambda_1,\ldots,\lambda_n)
\]
be a vector of positive primitive weights. The aggregation rule is required to satisfy three substantive conditions. First, \emph{same-scale normalization}: if all component ratings equal $r$, then the combined rating is also $r$. Second, \emph{recursive consistency}: aggregation is independent of grouping. For example, classical and rapid may be combined first into a subrating and then combined with blitz, provided the subrating carries the total weight of its components, and the final answer must be the same as aggregating all three ratings at once. Third, \emph{marginal Elo-strength consistency}: for two equally weighted coordinates, the ratio of their marginal contributions to the combined rating equals the ordinary Elo odds between them.

The main result, Theorem~\ref{thm:main}, shows that these requirements, together with mild regularity, identify a unique rule:
\begin{equation*}\label{eq:intro-combined}
C_\lambda(R)
=
400\log_{10}\!\left(
\frac{\sum_{i=1}^n\lambda_i\,10^{R_i/400}}{\sum_{i=1}^n\lambda_i}
\right).
\end{equation*}
Equivalently, each component rating is first converted into its Elo strength $10^{R_i/400}$, these strengths are averaged using the weights $\lambda_i$, and the result is converted back to the Elo scale. The denominator is essential: it ensures that a uniform profile $R_1=\cdots=R_n=r$ returns the rating $r$.

The setting connects three strands of literature. The Bradley--Terry and Elo models, together with Luce's choice axiom, provide the underlying log-odds geometry \citep{BradleyTerry1952,Elo1978,Luce1959}. Linear Tullock contest success functions produce the same ratio form in contest theory \citep{Tullock1980,Skaperdas1996}.

The paper also relates to existing practice in chess ratings. Most rating systems maintain separate ratings across time controls, including FIDE, US Chess, and major online platforms such as Chess.com and Lichess.org. Earlier proposals have explored incorporating rapid and blitz information into broader measures of playing strength \citep{Glickman1995,Sonas2002}. Systems such as the Universal Rating System (URS) instead aim to estimate a player's underlying ``classical'' strength using evidence from games played at different time controls \citep{URS2017}. The present paper studies a different object: not a target-format estimate, but a \emph{combined} strength measure across the chosen formats themselves. Given existing ratings already expressed on a common Elo scale, the paper asks which scalar aggregation rule preserves the Elo interpretation under natural consistency requirements, and derives the resulting representation theorem.

The remainder of the paper is organized as follows. Section~\ref{sec:model} sets up notation and states the axioms. Section~\ref{sec:characterization} proves the characterization theorem. Section~\ref{sec:properties} discusses the induced pairwise probabilities and basic properties of the combined rating. Section~\ref{sec:lottery-and-alternatives} compares the combined rating with random-format competitions and alternative scalar summaries. Section~\ref{sec:crb-application} applies the rule to classical, rapid, and blitz chess ratings. Section~\ref{sec:applications} discusses other applications and scope. Appendix~\ref{app:independence} proves independence of the substantive axioms.
\section{The setup}\label{sec:model}

The ordinary Elo expected score of a rating $a$ against a rating $b$ is
\begin{equation*}\label{eq:elo-E}
E(a,b)=\frac{1}{1+10^{(b-a)/400}}.
\end{equation*}
The corresponding odds are
\begin{equation}\label{eq:elo-odds}
\frac{E(a,b)}{1-E(a,b)}=10^{(a-b)/400}.
\end{equation}
Thus a 400-point rating difference corresponds to ten-to-one odds.

It is useful to write
\begin{equation*}\label{eq:q-def}
q(r)=10^{r/400},
\qquad
q^{-1}(s)=400\log_{10}s.
\end{equation*}
The quantity $q(r)$ is the Elo strength corresponding to rating $r$. The Elo expected score can then be written as
\[
E(a,b)=\frac{q(a)}{q(a)+q(b)}.
\]

A combined rating rule assigns to each positive weight vector $\lambda=(\lambda_1,\ldots,\lambda_n)\in(0,\infty)^n$ a function
\[
C_\lambda:\R^n\to\R.
\]
When player $R$ faces player $S$ under the same weights, the scalar Elo prediction generated by $C_\lambda$ is
\begin{equation*}\label{eq:combined-elo-prediction}
P_\lambda(R,S)
=
E(C_\lambda(R),C_\lambda(S))
=
\frac{1}{1+10^{(C_\lambda(S)-C_\lambda(R))/400}}.
\end{equation*}

Let
\[
\bar{\lambda}=\sum_{i=1}^n\lambda_i.
\]
For a nonempty subset $B\subseteq\{1,\ldots,n\}$, write $R_B=(R_i)_{i\in B}$ and $\lambda_B=(\lambda_i)_{i\in B}$, with the inherited coordinate order, and define
\[
\bar{\lambda}_B=\sum_{i\in B}\lambda_i.
\]

\subsection{Axioms}\label{sec:axioms}

The axioms are stated directly for the scalar Elo rating $C_\lambda$. The question is when several Elo ratings can be aggregated into a single rating that itself remains on the same Elo scale.

\begin{assumption}[Same-scale normalization]\label{ass:normalization}
For each positive weight vector $\lambda$,
\[
C_\lambda(r\one)=r
\qquad\text{for every }r\in\R.
\]
\end{assumption}

If all coordinates equal $r$, the combined rating is also $r$. This fixes the otherwise arbitrary additive level of the final Elo scale.

\begin{assumption}[Regularity and monotonicity]\label{ass:regularity}
For each dimension $n$, the map $(\lambda,R)\mapsto C_\lambda(R)$ is continuous on $(0,\infty)^n\times\R^n$. For each fixed $\lambda$, the function $C_\lambda$ is continuously differentiable and has strictly positive partial derivatives:
\[
\partial_i C_\lambda(R)>0
\qquad\text{for every }i\text{ and every }R\in\R^n.
\]
\end{assumption}

The positivity condition ensures that marginal contribution ratios are well-defined.

\begin{assumption}[Recursive consistency]\label{ass:recursive}
Let $\mathcal P=(B_1,\ldots,B_m)$ be an ordered partition of $\{1,\ldots,n\}$ into nonempty blocks. Then
\begin{equation}\label{eq:recursive}
C_\lambda(R)
=
C_{(\bar{\lambda}_{B_1},\ldots,\bar{\lambda}_{B_m})}
\left(
C_{\lambda_{B_1}}(R_{B_1}),\ldots,
C_{\lambda_{B_m}}(R_{B_m})
\right).
\end{equation}
\end{assumption}

Recursive consistency says that aggregation is path-independent. A consequence is replication invariance: splitting one coordinate of weight $\lambda_i$ into two identical subcoordinates with weights summing to $\lambda_i$ does not change the combined rating.

\begin{assumption}[Marginal Elo-strength consistency]\label{ass:marginal-elo}
For two equally weighted coordinates, the ratio of marginal contributions equals the ordinary Elo odds of one rating against the other. Formally, for all $x,y\in\R$,
\begin{equation*}\label{eq:marginal-axiom}
\frac{\partial_x C_{(1,1)}(x,y)}{\partial_y C_{(1,1)}(x,y)}
=
\frac{E(x,y)}{1-E(x,y)}.
\end{equation*}
\end{assumption}

Using \eqref{eq:elo-odds}, the right-hand side is $10^{(x-y)/400}$. Thus equal ratings have equal marginal influence; a 400-point higher coordinate has ten times the marginal influence of the lower coordinate. This is not an additional probability model. It is a local scale-selection condition: marginal importance is proportional to Elo strength.

\begin{remark}[What the marginal axiom fixes]
In the equal-weight binary case, Assumption~\ref{ass:marginal-elo} fixes a gradient direction. It implies that
\[
C_{(1,1)}(x,y)=h\bigl(10^{x/400}+10^{y/400}\bigr)
\]
for some strictly increasing function $h$. Same-scale normalization then identifies $h$. Hence the marginal axiom identifies the Elo-strength geometry, while normalization fixes the calibration of the final rating.
\end{remark}

\subsection{Immediate consequences of recursion}\label{sec:recursive-consequences}

No anonymity axiom is imposed separately. With recursive consistency stated for arbitrary ordered partitions, joint relabeling invariance follows from recursion and normalization. No weight-scale axiom is imposed either; common rescaling of all weights is also a consequence.

\begin{lemma}[Joint relabeling invariance]\label{lem:anonymity}
Under Assumptions~\ref{ass:normalization} and~\ref{ass:recursive}, for every permutation $\pi$ of $\{1,\ldots,n\}$,
\[
C_\lambda(R)
=
C_{(\lambda_{\pi(1)},\ldots,\lambda_{\pi(n)})}
(R_{\pi(1)},\ldots,R_{\pi(n)}).
\]
\end{lemma}

\begin{proof}
Apply recursive consistency to the ordered singleton partition
\[
\mathcal P=(\{\pi(1)\},\ldots,\{\pi(n)\}).
\]
For singleton blocks, same-scale normalization gives
\[
C_{(\lambda_{\pi(j)})}(R_{\pi(j)})=R_{\pi(j)}.
\]
Substitution into \eqref{eq:recursive} gives the claim.
\end{proof}

\begin{lemma}[Weight-scale invariance]\label{lem:scale-invariance}
Under Assumptions~\ref{ass:normalization}--\ref{ass:recursive},
\[
C_{\alpha\lambda}(R)=C_\lambda(R)
\qquad\text{for every }\alpha>0.
\]
\end{lemma}

\begin{proof}
Fix $\lambda\in(0,\infty)^n$, $R\in\R^n$, and a positive integer $k$. Form an expanded $kn$-coordinate profile indexed by $(i,j)$, with
\[
\widehat\lambda_{i,j}=\lambda_i,
\qquad
\widehat R_{i,j}=R_i,
\qquad
1\le i\le n,
\quad
1\le j\le k.
\]
First partition the expanded coordinates by original coordinate. The block for coordinate $i$ has total weight $k\lambda_i$ and all ratings equal to $R_i$, so recursive consistency and normalization give
\[
C_{\widehat\lambda}(\widehat R)=C_{k\lambda}(R).
\]
Second, partition the expanded coordinates into $k$ full copies of the original profile. Each block rating is $C_\lambda(R)$, and each block has total weight $\bar{\lambda}=\sum_i\lambda_i$. Hence
\[
C_{\widehat\lambda}(\widehat R)
=
C_{(\bar{\lambda},\ldots,\bar{\lambda})}
\bigl(C_\lambda(R),\ldots,C_\lambda(R)\bigr)
=
C_\lambda(R).
\]
Thus $C_{k\lambda}(R)=C_\lambda(R)$ for every positive integer $k$. The same identity for positive rational scale factors follows by applying the integer case twice. Continuity in weights then extends the identity to every $\alpha>0$.
\end{proof}

\section{Characterization}\label{sec:characterization}

\begin{theorem}[Recursive Elo characterization]\label{thm:main}
Suppose Assumptions~\ref{ass:normalization}--\ref{ass:marginal-elo} hold. Then, for every $n\ge1$, every $\lambda\in(0,\infty)^n$, and every $R\in\R^n$,
\begin{equation}\label{eq:main-C}
C_\lambda(R)
=
400\log_{10}\left(
\frac{\sum_{i=1}^n\lambda_i10^{R_i/400}}{\sum_{i=1}^n\lambda_i}
\right).
\end{equation}
Equivalently,
\begin{equation*}\label{eq:main-Q}
q\bigl(C_\lambda(R)\bigr)
=
\frac{\sum_{i=1}^n\lambda_iq(R_i)}{\sum_{i=1}^n\lambda_i}.
\end{equation*}
Conversely, the rule \eqref{eq:main-C} satisfies Assumptions~\ref{ass:normalization}--\ref{ass:marginal-elo}. The weights are identified up to multiplication by a common positive scalar; under the normalization $\sum_i\lambda_i=1$, they are unique.
\end{theorem}

\begin{proof}
The proof uses $q(r)=10^{r/400}$ and $q^{-1}(s)=400\log_{10}s$.

\medskip
\noindent\text{Step 1: the equal-weight binary rule.}
Let
\[
M(x,y)=C_{(1,1)}(x,y).
\]
By Assumption~\ref{ass:marginal-elo},
\[
\frac{M_x(x,y)}{M_y(x,y)}
=10^{(x-y)/400}
=\frac{q(x)}{q(y)}.
\]
Let $V(x,y)=q(x)+q(y)$. Since $V_x/V_y=q(x)/q(y)$, the gradients of $M$ and $V$ are everywhere parallel. Along any level curve of $V$,
\[
\frac{d}{dx}M(x,y(x))=M_x+M_y y'(x)=0.
\]
The level curves of $V$ are connected, so $M$ is constant on them. Therefore
\[
M(x,y)=h(q(x)+q(y))
\]
for some strictly increasing function $h$. Same-scale normalization gives $M(r,r)=r$, hence
\[
h(2q(r))=r
\qquad\text{for all }r.
\]
Since $q(\R)=(0,\infty)$,
\[
h(s)=q^{-1}\left(\frac{s}{2}\right).
\]
Thus
\begin{equation}\label{eq:binary-equal}
C_{(1,1)}(x,y)
=
q^{-1}\left(\frac{q(x)+q(y)}{2}\right).
\end{equation}

\medskip
\noindent\text{Step 2: equal weights with dyadic replication.}
For every $m=2^k$,
\begin{equation}\label{eq:equal-power-two}
C_{\one_m}(R_1,\ldots,R_m)
=
q^{-1}\left(\frac{1}{m}\sum_{i=1}^m q(R_i)\right).
\end{equation}
The case $m=1$ follows from normalization and the case $m=2$ from \eqref{eq:binary-equal}. If \eqref{eq:equal-power-two} holds for $m$, partition $2m$ equal-weight coordinates into two blocks of size $m$. Recursive consistency gives
\[
C_{\one_{2m}}(R)
=
C_{(m,m)}\left(
C_{\one_m}(R_1,\ldots,R_m),
C_{\one_m}(R_{m+1},\ldots,R_{2m})
\right).
\]
By Lemma~\ref{lem:scale-invariance}, $C_{(m,m)}=C_{(1,1)}$. Applying the induction hypothesis and \eqref{eq:binary-equal} proves \eqref{eq:equal-power-two} for $2m$.

\medskip
\noindent\text{Step 3: dyadic weights.}
Let $w=(w_1,\ldots,w_n)$ be a positive dyadic weight vector with $\sum_iw_i=1$. Then there exist positive integers $m_i$ and a power of two $m=2^k$ such that $w_i=m_i/m$ and $\sum_i m_i=m$. Repeat coordinate $R_i$ exactly $m_i$ times to form an expanded equal-weight profile $\widehat R$. By Step 2,
\[
C_{\one_m}(\widehat R)
=
q^{-1}\left(\sum_{i=1}^n w_iq(R_i)\right).
\]
Partitioning the expanded profile into blocks of identical copies and applying recursive consistency gives
\[
C_{\one_m}(\widehat R)=C_{(m_1,\ldots,m_n)}(R).
\]
By weight-scale invariance,
\begin{equation}\label{eq:dyadic}
C_w(R)=q^{-1}\left(\sum_{i=1}^n w_iq(R_i)\right).
\end{equation}

\medskip
\noindent\text{Step 4: arbitrary positive weights.}
For arbitrary $\lambda$, set $w_i=\lambda_i/\bar{\lambda}$. By Lemma~\ref{lem:scale-invariance}, $C_\lambda=C_w$. Approximate $w$ by positive dyadic probability vectors $w^{(k)}$. Continuity in weights and \eqref{eq:dyadic} yield
\[
C_w(R)
=\lim_{k\to\infty}C_{w^{(k)}}(R)
=q^{-1}\left(\sum_{i=1}^n w_iq(R_i)\right),
\]
which is \eqref{eq:main-C}.

\medskip
\noindent\text{Step 5: converse and uniqueness of weights.}
Assume $C_\lambda$ is defined by \eqref{eq:main-C}. Same-scale normalization is immediate. The rule is continuous, continuously differentiable, and strictly increasing because
\begin{equation}\label{eq:derivative}
\partial_j C_\lambda(R)=
\frac{\lambda_jq(R_j)}{\sum_{i=1}^n\lambda_iq(R_i)}>0.
\end{equation}
Recursive consistency holds on the strength scale: for a block $B$,
\[
q\bigl(C_{\lambda_B}(R_B)\bigr)
=
\frac{\sum_{i\in B}\lambda_iq(R_i)}{\bar{\lambda}_B}.
\]
A weighted average of these block strengths with block weights $\bar{\lambda}_B$ equals the direct strength average $\sum_i\lambda_iq(R_i)/\bar{\lambda}$. Converting back by $q^{-1}$ gives recursion.

For two equal weights, differentiating \eqref{eq:binary-equal} gives
\[
\partial_x C_{(1,1)}(x,y)=\frac{q(x)}{q(x)+q(y)},
\qquad
\partial_y C_{(1,1)}(x,y)=\frac{q(y)}{q(x)+q(y)}.
\]
Their ratio is $q(x)/q(y)=10^{(x-y)/400}=E(x,y)/(1-E(x,y))$, so marginal Elo-strength consistency holds.

Finally, \eqref{eq:main-C} depends only on normalized weights. Conversely, if two normalized weight vectors $w$ and $v$ produce the same rule for all $R$, then
\[
\sum_i w_iq(R_i)=\sum_i v_iq(R_i)
\qquad\text{for all }R.
\]
Set all coordinates except $i$ equal to zero and let $R_i=r$. Since $q(0)=1$,
\[
w_iq(r)+(1-w_i)=v_iq(r)+(1-v_i)
\qquad\text{for all }r.
\]
Choosing any $r\ne0$ gives $w_i=v_i$. Hence normalized weights are unique.
\end{proof}

\section{Properties}\label{sec:properties}

\subsection{Induced pairwise probabilities}\label{subsec:probabilities}

The characterized combined rating induces the scalar Elo probability
\begin{equation*}\label{eq:prob-final}
P_\lambda(R,S)
=
\frac{1}{1+10^{(C_\lambda(S)-C_\lambda(R))/400}}.
\end{equation*}
Using Theorem~\ref{thm:main}, this becomes
\begin{equation}\label{eq:prob-strength}
P_\lambda(R,S)
=
\frac{\sum_i\lambda_i10^{R_i/400}}
{\sum_i\lambda_i10^{R_i/400}+\sum_i\lambda_i10^{S_i/400}}.
\end{equation}
Thus the final probability has the Bradley--Terry form, with aggregate strength equal to the weighted average of component Elo strengths.

Formula \eqref{eq:prob-strength} also has a Luce-choice interpretation. Give player $R$'s coordinate $i$ mass $\lambda_i10^{R_i/400}$, and similarly for player $S$. Put all player-coordinate masses into one pool. A component is selected with probability proportional to its mass, and a player wins if one of that player's components is selected. Then the probability that $R$ wins is exactly \eqref{eq:prob-strength}.

\begin{proposition}[Uniqueness under component-level pooling]\label{prop:component}
Suppose
\[
P(R,S)=
\frac{\sum_i\lambda_i10^{R_i/400}}
{\sum_i\lambda_i10^{R_i/400}+\sum_i\lambda_i10^{S_i/400}},
\]
and suppose there exists a scalar Elo rating $C$ such that
\[
P(R,S)=\frac{1}{1+10^{(C(S)-C(R))/400}}
\]
for all rating vectors $R,S$. Then
\[
C(R)=400\log_{10}\left(\sum_i\lambda_i10^{R_i/400}\right)+K
\]
for some constant $K$. If same-scale normalization is imposed, then $C$ is given by \eqref{eq:main-C}.
\end{proposition}

\begin{proof}
Let $H(R)=\sum_i\lambda_i10^{R_i/400}$. The pooling probability gives odds $H(R)/H(S)$, while the scalar Elo formula gives odds $10^{(C(R)-C(S))/400}$. Hence
\[
\frac{10^{C(R)/400}}{H(R)}=
\frac{10^{C(S)/400}}{H(S)}
\]
for all $R,S$. Therefore $10^{C(R)/400}=AH(R)$ for some $A>0$, which gives the first claim. Same-scale normalization gives $A=(\sum_i\lambda_i)^{-1}$.
\end{proof}

\begin{proposition}[Basic properties]\label{prop:basic}
Let $C_\lambda$ be given by \eqref{eq:main-C}. Then:
\begin{enumerate}[label=(\roman*)]
\item \textup{Translation equivariance.} For every $\delta\in\R$,
\[
C_\lambda(R+\delta\one)=C_\lambda(R)+\delta.
\]
\item \textup{Internality.} For every $R\in\R^n$,
\[
\min_iR_i\le C_\lambda(R)\le\max_iR_i.
\]
\item \textup{Endogenous marginal weights.}
\[
\frac{\partial C_\lambda(R)}{\partial R_j}
=
\frac{\lambda_jq(R_j)}{\sum_i\lambda_iq(R_i)},
\qquad
\sum_{j=1}^n\frac{\partial C_\lambda(R)}{\partial R_j}=1.
\]
\end{enumerate}
\end{proposition}

\begin{proof}
Translation equivariance follows from $q(R_i+\delta)=q(\delta)q(R_i)$. Internality follows because $q(C_\lambda(R))$ is a weighted arithmetic mean of the positive numbers $q(R_i)$, and $q^{-1}$ is increasing. The derivative formula is \eqref{eq:derivative}; summing over $j$ gives one.
\end{proof}

The last property shows that primitive weights and marginal weights are different. The primitive weights $\lambda_i$ are fixed policy parameters, while the effective marginal influence of coordinate $j$ is proportional to $\lambda_jq(R_j)$.

\begin{remark}[Zero-weight limits]
The axioms are stated for strictly positive weights. For the characterized rule, zero weights can be admitted by limits: a coordinate with zero weight simply drops out of the strength average. The representation theorem itself is stated on the positive-weight domain.
\end{remark}

\section{Comparisons}\label{sec:lottery-and-alternatives}

\subsection{Random-format competitions}\label{sec:lottery}

A \emph{random-format competition} is a competition in which a format is first selected according to exogenous probabilities and the game is then evaluated using the rating for the selected format. If format $i$ is selected with probability $\pi_i$, and if
\[
p_i(R,S)=\frac{10^{R_i/400}}{10^{R_i/400}+10^{S_i/400}}
\]
is the format-specific Elo expected score of $R$ against $S$, then the expected score for the random-format competition is
\begin{equation*}\label{eq:format-lottery}
P_\pi(R,S)=\sum_{i=1}^n\pi_i p_i(R,S).
\end{equation*}
This construction answers a different question from the combined rating. It averages probabilities after the format has been selected. The combined rating averages strengths first and then applies the Elo formula:
\begin{equation*}\label{eq:PC-strength}
P_C(R,S)=
\frac{\sum_i\lambda_i10^{R_i/400}}
{\sum_i\lambda_i10^{R_i/400}+\sum_i\lambda_i10^{S_i/400}}.
\end{equation*}

The two expressions can coincide in special cases, but not in general. The combined-rating probability can be written as
\begin{equation}\label{eq:endogenous-weights}
P_C(R,S)=\sum_i w_i(R,S)p_i(R,S),
\end{equation}
where
\begin{equation}\label{eq:wi}
w_i(R,S)=
\frac{\lambda_i(10^{R_i/400}+10^{S_i/400})}
{\sum_j\lambda_j(10^{R_j/400}+10^{S_j/400})}.
\end{equation}
The weights $w_i(R,S)$ are endogenous matchup weights. Even when the primitive weights $\lambda_i$ are equal, the weights in \eqref{eq:endogenous-weights} generally are not equal to $1/n$.

\begin{proposition}[A random-format cycle]\label{prop:cycle}
A random-format lottery need not admit a scalar-rating representation. Under uniform selection over three formats, there exist profiles $X,Y,Z$ such that $X$ is favored over $Y$, $Y$ is favored over $Z$, and $Z$ is favored over $X$.
\end{proposition}

\begin{proof}
Let
\[
X=(2800,2400,2000),\quad
Y=(2400,2000,2800),\quad
Z=(2000,2800,2400).
\]
Under uniform format selection, comparing $X$ to $Y$ gives rating differences $(400,400,-800)$. Therefore
\[
P_\unif(X\text{ beats }Y)
=
\frac13\left(\frac{10}{11}+\frac{10}{11}+\frac{1}{101}\right)
=
\frac{677}{1111}>
\frac12.
\]
By cyclic symmetry, $P_\unif(Y\text{ beats }Z)>1/2$ and $P_\unif(Z\text{ beats }X)>1/2$. No scalar rating can represent these strict comparisons.
\end{proof}

For the same three profiles, the equal-weight combined rating assigns the same scalar rating to all three players, because each profile contains the same strength values $\{10^7,10^6,10^5\}$, only in a different order. The combined-rating probabilities are therefore all $1/2$. This illustrates that a random-format lottery and a combined-strength rating are different objects.

In practice, the two constructions are often close when the coordinate ratings of the players do not differ too sharply across formats. In such cases the matchup-dependent endogenous weights $w_i(R,S)$ remain close to the exogenous lottery probabilities $\pi_i$, so the combined-rating probability approximates the corresponding random-format expected score.

\subsection{Alternative scalar summaries}\label{sec:other-summaries}

One simple alternative is the weighted arithmetic average of ratings,
\[
\bar R(R)=\sum_i\alpha_iR_i,
\qquad
\alpha_i>0,
\qquad
\sum_i\alpha_i=1.
\]
On the strength scale, this is the weighted geometric mean
\[
10^{\bar R(R)/400}
=
\prod_i10^{\alpha_iR_i/400}.
\]
It is recursively consistent if the rating scale itself is treated as the averaging scale, but it violates marginal Elo-strength consistency: in the equal-weight two-coordinate case, its marginal ratio is always one, not the Elo odds.

More generally, one may use a power mean of Elo strengths,
\[
G_p(q)=\left(\frac{\sum_i\lambda_iq_i^p}{\sum_i\lambda_i}\right)^{1/p},
\qquad p\ne0,
\]
with the geometric mean obtained as $p\to0$. The associated rating is $400\log_{10}G_p(q(R))$. These means are recursively consistent, but their equal-weight binary marginal ratio is
\[
\left(\frac{10^{x/400}}{10^{y/400}}\right)^p
=10^{p(x-y)/400}.
\]
Marginal Elo-strength consistency requires $p=1$. Other values of $p$ encode different assumptions about substitutability across coordinates.

\subsection{Other applications and scope}\label{sec:applications}

The same framework can aggregate role-specific ratings. In chess, White and Black are strategically different roles, so one could maintain
\[
R=(R_W,R_B).
\]
The displayed role-aggregate rating with equal weights would be
\[
C(R_W,R_B)=400\log_{10}\left(
\frac{10^{R_W/400}+10^{R_B/400}}{2}
\right).
\]
If player $A$ has White against player $B$, the role-specific expected score for $A$ is
\[
E_A=
\frac{10^{R_W(A)/400}}
{10^{R_W(A)/400}+10^{R_B(B)/400}}.
\]
A standard Elo-style update with factor $K$ updates only the played role coordinates:
\[
R_W(A)^+=R_W(A)+K(s_A-E_A),
\qquad
R_B(B)^+=R_B(B)+K(E_A-s_A),
\]
where $s_A\in\{0,1/2,1\}$ is White's score. The unplayed coordinates remain unchanged, and each player's displayed scalar rating is recomputed from the role coordinates.

\section{Application: combining classical, rapid, and blitz ratings}\label{sec:crb-application}

For a chess application, take the three coordinates to be classical, rapid, and blitz:
\[
R=(R_c,R_r,R_b).
\]
With weights $\lambda=(\lambda_c,\lambda_r,\lambda_b)$, the combined rating is
\begin{equation}\label{eq:crb-weighted}
C_{c,r,b}(R)
=
400\log_{10}\left(
\frac{\lambda_c10^{R_c/400}+\lambda_r10^{R_r/400}+\lambda_b10^{R_b/400}}
{\lambda_c+\lambda_r+\lambda_b}
\right).
\end{equation}
This is the player's \emph{Combined Rating}: a single Elo rating corresponding to combined strength across classical, rapid, and blitz.

The construction is simple in operational terms. Each format rating is translated into an Elo strength
\[
q_f=10^{R_f/400},
\]
the three strengths are averaged, and the result is converted back using $400\log_{10}$. The division by three in the equal-weight case is the normalization that makes $C(R,R,R)=R$.

The interpretation is equally important. Formula \eqref{eq:crb-weighted} measures combined classical--rapid--blitz strength. It is not a replacement for the classical rating when a classical game is being predicted, and it is not an estimate of what a player's classical rating would be after using rapid and blitz results as auxiliary evidence. In particular, a URS-type system, when understood as measuring ``classical'' strength using information from non-classical games, has a different target from the combined rating here. The present construction treats the time controls as separate dimensions and asks for one Elo-style aggregate of those dimensions.

Equivalently, the Combined Rating answers the question:
\begin{quote}
If we want a single Elo rating that summarizes how strong this player is overall across classical, rapid, and blitz, what should that rating be?
\end{quote}
It assigns each player a single scalar rating, permits standard Elo predictions from rating differences, and is nonlinear in the original format ratings. The resulting predictions are often close to those of a corresponding random-format lottery when the players' ratings do not differ too sharply across formats.

As a numerical illustration, consider the classical, rapid, and blitz ratings of Carlsen $(A)$ and Gukesh $(B)$:
\[
A=(2840,2832,2869),
\qquad
B=(2732,2692,2646).
\]
With equal weights,
\[
C(A)=2847.74,
\qquad
C(B)=2693.52.
\]
The combined rating difference is
\[
C(A)-C(B)=154.22.
\]
Applying the ordinary Elo logistic formula to this difference gives
\[
\frac{1}{1+10^{-154.22/400}}=0.7084.
\]
Thus, in this combined-strength model, $A$ is about a $70.84\%$ favorite against $B$ in an overall multi-format comparison. This example shows the basic use of the Combined Rating: it converts three separate ratings into one number that can be used in the standard Elo prediction formula.

For the same ratings, the format-specific expected scores are
\[
p_c=0.6506,
\qquad
p_r=0.6912,
\qquad
p_b=0.7831.
\]
A uniform random-format lottery would therefore give
\[
P_\unif(A,B)=\frac{0.6506+0.6912+0.7831}{3}=0.7083.
\]
The equal-weight combined-rating prediction is $0.7084$, so in this example the two predictions differ by only about $0.0001$. This closeness is not automatic; here the total strength masses
\[
q(A_f)+q(B_f),\qquad f\in\{c,r,b\},
\]
are relatively balanced across formats, so the endogenous weights in \eqref{eq:wi} are close to the uniform lottery weights. The approximation would be poorer in matchups where one format has a much larger total strength mass and also a very different format-specific expected score.

The choice of weights is a policy choice about the object being measured. Once the weights are chosen, Theorem~\ref{thm:main} fixes the unique normalized combined Elo rating satisfying the axioms.

Table~\ref{tab:top20-combined} gives an illustrative ranking using the live rating values reported by 2700chess.com for April 19, 2026. The table shows that the Combined Rating can differ substantially from the classical ranking, because it incorporates rapid and blitz strength alongside classical strength under equal weighting across the three formats.

\begin{table}[H]
\centering
\small
\setlength{\tabcolsep}{3.5pt}
\begin{tabular*}{\textwidth}{@{\extracolsep{\fill}}rrlrrrr@{}}
\toprule
Combined rank & Classical rank & Name & Combined & Classical & Rapid & Blitz \\
\midrule
1 & 1 & Carlsen, Magnus & 2848 & 2840 & 2832 & 2869 \\
2 & 2 & Nakamura, Hikaru & 2795 & 2792 & 2742 & 2838 \\
3 & 8 & Firouzja, Alireza & 2771 & 2759 & 2755 & 2796 \\
4 & 4 & Abdusattorov, Nodirbek & 2760 & 2780 & 2703 & 2785 \\
5 & 11 & Erigaisi, Arjun & 2757 & 2751 & 2741 & 2776 \\
6 & 3 & Caruana, Fabiano & 2757 & 2788 & 2727 & 2749 \\
7 & 9 & So, Wesley & 2756 & 2754 & 2705 & 2798 \\
8 & 19 & Nepomniachtchi, Ian & 2741 & 2729 & 2726 & 2765 \\
9 & 24 & Vachier-Lagrave, Maxime & 2739 & 2717 & 2735 & 2761 \\
10 & 21 & Aronian, Levon & 2731 & 2724 & 2730 & 2740 \\
11 & 5 & Sindarov, Javokhir & 2728 & 2776 & 2727 & 2662 \\
12 & 10 & Wei, Yi & 2727 & 2753 & 2726 & 2698 \\
13 & 13 & Duda, Jan-Krzysztof & 2724 & 2739 & 2683 & 2743 \\
14 & 59 & Dubov, Daniil & 2721 & 2654 & 2687 & 2792 \\
15 & 32 & Fedoseev, Vladimir & 2719 & 2703 & 2690 & 2756 \\
16 & 6 & Giri, Anish & 2713 & 2767 & 2689 & 2666 \\
17 & 22 & Nihal Sarin & 2712 & 2723 & 2689 & 2723 \\
18 & 77 & Artemiev, Vladislav & 2708 & 2641 & 2735 & 2733 \\
19 & 44 & Leko, Peter & 2707 & 2676 & 2700 & 2738 \\
20 & 15 & Praggnanandhaa R & 2700 & 2733 & 2663 & 2698 \\
\bottomrule
\end{tabular*}
\caption{Illustrative equal-weight Combined Ratings for the top 20 players as of April 19, 2026. Combined ratings are rounded to the nearest point.}
\label{tab:top20-combined}
\end{table}

\appendix

\section{Independence of the substantive axioms}\label{app:independence}

This appendix explains which assumptions do substantive work in the characterization. Anonymity and weight-scale invariance are not included among the axioms because they follow from Lemmata~\ref{lem:anonymity} and~\ref{lem:scale-invariance}. The regularity and monotonicity assumptions are retained throughout as technical conditions. Thus the independence result concerns the three substantive axioms: same-scale normalization, recursive consistency, and marginal Elo-strength consistency.

\begin{proposition}[Independence of the substantive axioms]\label{prop:independence}
Retaining the regularity and monotonicity requirements, each of Assumptions~\ref{ass:normalization}, \ref{ass:recursive}, and \ref{ass:marginal-elo} is independent of the other two.
\end{proposition}

\begin{proof}
We give one counterexample for each omitted assumption.

\medskip
\noindent\text{1. Marginal Elo-strength consistency is independent.}
Omit Assumption~\ref{ass:marginal-elo} and define
\[
C_\lambda(R)=\sum_{i=1}^n\frac{\lambda_i}{\bar{\lambda}}R_i.
\]
This weighted arithmetic mean of ratings is continuous, continuously differentiable, strictly increasing, and satisfies same-scale normalization. It also satisfies recursive consistency: block averages recombine to the original weighted average. But for two equal weights,
\[
\partial_x C_{(1,1)}(x,y)=\partial_y C_{(1,1)}(x,y)=\frac12,
\]
so the marginal ratio is always one, not $10^{(x-y)/400}$ unless $x=y$. Hence the marginal axiom is independent.

\medskip
\noindent\text{2. Recursive consistency is independent.}
Omit Assumption~\ref{ass:recursive}. Define
\[
C_\lambda(R)=
\begin{cases}
q^{-1}\left(\dfrac{\sum_i\lambda_iq(R_i)}{\sum_i\lambda_i}\right),
& n=1\text{ or }n=2,\\[1.2em]
\dfrac{\sum_i\lambda_iR_i}{\sum_i\lambda_i},
& n\ge3.
\end{cases}
\]
This rule is continuous, continuously differentiable, strictly increasing, and satisfies same-scale normalization. It also satisfies the marginal axiom because the two-dimensional equal-weight branch is exactly the strength-average rule.

Recursive consistency fails. With equal weights and $R=(0,400,400)$, direct three-dimensional aggregation uses the arithmetic-rating branch and gives
\[
C_{(1,1,1)}(0,400,400)=\frac{800}{3}.
\]
If the first two coordinates are combined first, their block rating is
\[
C_{(1,1)}(0,400)=q^{-1}\left(\frac{1+10}{2}\right)=q^{-1}\left(\frac{11}{2}\right).
\]
Combining this block with the third coordinate using block weights $(2,1)$ gives strength
\[
\frac{2\cdot(11/2)+1\cdot10}{3}=7,
\]
and hence rating $q^{-1}(7)=400\log_{10}7$, which is not $800/3$. Hence recursive consistency is independent.

\medskip
\noindent\text{3. Same-scale normalization is independent.}
Omit Assumption~\ref{ass:normalization}. Let
\[
w_i=\frac{\lambda_i}{\bar{\lambda}},
\qquad
H(w)=-\sum_{i=1}^n w_i\log w_i,
\]
where $H(w)$ is Shannon entropy, with natural logarithm. Fix $\eta>0$, and define
\[
Q_\lambda(R)=\sum_{i=1}^n w_iq(R_i)+\eta H(w),
\qquad
C_\lambda(R)=q^{-1}\bigl(Q_\lambda(R)\bigr).
\]
This rule is well-defined, continuous, continuously differentiable in $R$, and strictly increasing in each coordinate, because
\[
\partial_j C_\lambda(R)
=
\frac{w_jq(R_j)}{\sum_i w_iq(R_i)+\eta H(w)}>0.
\]

It satisfies marginal Elo-strength consistency. For two equal weights,
\[
Q_{(1,1)}(x,y)=\frac{q(x)+q(y)}{2}+\eta\log2,
\]
and the entropy term is constant in $x$ and $y$. Therefore the marginal ratio is $q(x)/q(y)=10^{(x-y)/400}$.

It also satisfies recursive consistency. For a partition $\mathcal P$, write
\[
W_B=\sum_{i\in B}w_i,
\qquad
w_i^B=\frac{w_i}{W_B}\quad(i\in B).
\]
The block-level strength is
\[
Q_{\lambda_B}(R_B)
=
\sum_{i\in B}w_i^Bq(R_i)+\eta H(w^B).
\]
Aggregating block-level ratings using normalized block weights $W_B$ gives strength
\begin{align*}
&\sum_{B\in\mathcal P}W_B Q_{\lambda_B}(R_B)+\eta H(W)\\
&\qquad=
\sum_i w_iq(R_i)
+
\eta\left(\sum_{B\in\mathcal P}W_BH(w^B)+H(W)\right).
\end{align*}
Shannon entropy satisfies the grouping identity
\[
H(w)=H(W)+\sum_{B\in\mathcal P}W_BH(w^B),
\]
so the recursively aggregated strength equals $Q_\lambda(R)$.

Same-scale normalization fails. If all coordinates equal $r$, then
\[
Q_\lambda(r\one)=q(r)+\eta H(w).
\]
For any nontrivial positive weight vector with $n\ge2$, $H(w)>0$. Hence $C_\lambda(r\one)>r$. Thus normalization is independent.
\end{proof}

\begin{remark}[Role of regularity]
The counterexamples in Proposition~\ref{prop:independence} preserve the stated regularity and monotonicity requirements. Continuity in weights is used in Theorem~\ref{thm:main} to extend the formula from dyadic rational weights to arbitrary positive weights. Differentiability and positive partial derivatives make the marginal ratio well-defined and allow the marginal axiom to be translated into a statement about level curves.
\end{remark}


\begin{thebibliography}{99}

\bibitem[Bradley and Terry(1952)]{BradleyTerry1952}
Bradley, Ralph Allan, and Milton E. Terry. 1952.
``Rank Analysis of Incomplete Block Designs: I. The Method of Paired Comparisons.''
\emph{Biometrika} 39(3/4): 324--345.
\url{https://doi.org/10.1093/biomet/39.3-4.324}.


\bibitem[Elo(1978)]{Elo1978}
Elo, Arpad E. 1978.
\emph{The Rating of Chessplayers, Past and Present}.
New York: Arco.

\bibitem[Glickman(1995)]{Glickman1995}
Glickman, Mark E. 1995.
``A Comprehensive Guide to Chess Ratings.''
\emph{American Chess Journal} 3: 59--102.
\url{https://www.glicko.net/research/acjpaper.pdf}.

\bibitem[Luce(1959)]{Luce1959}
Luce, R. Duncan. 1959.
\emph{Individual Choice Behavior: A Theoretical Analysis}.
New York: Wiley.

\bibitem[Skaperdas(1996)]{Skaperdas1996}
Skaperdas, Stergios. 1996.
``Contest Success Functions.''
\emph{Economic Theory} 7(2): 283--290.
\url{https://doi.org/10.1007/BF01213906}.


\bibitem[Sonas(2002)]{Sonas2002}
Sonas, Jeff. 2002.
``The Sonas Rating Formula -- Better than Elo?''
ChessBase, October 22.
\url{https://en.chessbase.com/post/the-sonas-rating-formula-better-than-elo}.
Accessed May 8, 2026.

\bibitem[Tullock(1980)]{Tullock1980}
Tullock, Gordon. 1980.
``Efficient Rent Seeking.''
In \emph{Toward a Theory of the Rent-Seeking Society}, edited by James M. Buchanan, Robert D. Tollison, and Gordon Tullock, 97--112.
College Station: Texas A\&M University Press.

\bibitem[Universal Rating System(2017)]{URS2017}
Universal Rating System. 2017.
``Universal Rating System.''
\url{http://universalrating.com/about-us.php}.
Accessed May 8, 2026.

\end{thebibliography}
\end{document}